\documentclass[conference]{IEEEtran}

\usepackage{cite}      

\usepackage[usenames,dvipsnames,svgnames,table]{xcolor}
\usepackage{graphicx}  

\usepackage{psfrag}    

\usepackage{url}       
\usepackage{tikz}

\usepackage{amsmath}   
\usepackage{amssymb}

\newcommand{\shorten}[1]{}
\newtheorem{proposition}{Proposition}
\newtheorem{theorem}{Theorem}
\newtheorem{definition}{Definition}

\newtheorem{example}{Example}
\newtheorem{problem}{Problem}
\newcommand{\signed}%
    {{\unskip\nobreak\hfill\penalty50
      \hskip2em\hbox{}\nobreak\hfil $\blacksquare$
      \parfillskip=0pt \finalhyphendemerits=0 \par}}
\newenvironment{proof}[1]
    {
    \bf{Proof:}\rm{\noindent{#1 }}\ignorespaces
    }
    {\signed\addvspace\medskipamount}

\begin{document}

\title{Families of Optimal Binary Non-MDS Erasure Codes}

%
\author{\authorblockN{
Danilo Gligoroski and Katina Kralevska}
\authorblockA{Department of Telematics, Faculty of Information
Technology, Mathematics and Electrical Engineering, \\ Norwegian University of Science and Technology, Trondheim, Norway,\\ Email:
\{danilog, katinak\}@item.ntnu.no}

}


\maketitle

\begin{abstract}
We introduce a definition for \emph{Families of Optimal Binary Non-MDS Erasure Codes} for $[n, k]$ codes over $GF(2)$, and propose an algorithm for finding those families by using hill climbing techniques over Balanced XOR codes. Due to the hill climbing search, those families of codes have always better decoding probability than the codes generated in a typical Random Linear Network Coding scenario, i.e., random linear codes. We also show a surprising result that for small values of $k$, the decoding probability of our codes in $GF(2)$ is very close to the decoding probability of the codes obtained by Random Linear Network Coding but in the higher finite field $GF(4)$.
\end{abstract}


%
\IEEEpeerreviewmaketitle

\section{Introduction}

In the fast approaching Zettabyte Era \cite{CiscoVNI2013} the erasure codes will become the most important codes among all coding techniques. That is mostly due to two factors: 1. The global communications will be almost exclusively based on the packet switching paradigm, where the recovery from packet losses is addressed efficiently by erasure codes; 2. Storage systems will have capacities of hundreds of exabytes, and will have to tolerate and recover efficiently from multiple disk failures.

According to the rate of redundancy that is used, the erasure codes are divided in two classes: 1. Optimal or very close to optimal ones, known as Maximum Distance Separable (MDS) Codes \cite{bok:MW}, almost-MDS (AMDS) \cite{journals/dcc/Boer96} and near-MDS codes (NMDS) \cite{journals/jgeom/Dodunekov95}, and 2: Suboptimal or non-MDS codes \cite{Akhlaghi,conf/stoc/LubyMSSS97,conf/fast/Hafner05,journals/tos/HuangCL13,conf/mss/GreenanLW10}.

Reed-Solomon codes~\cite{reed60polynomial} are a well known class of MDS codes that provide a general technique for construction of MDS codes. However, these codes are defined in higher finite fields and they can be very computationally demanding. That is the main reason for series of research efforts to find codes that work just in the simplest finite field $GF(2)$ where the operations are bitwise exclusive-or (XOR) operations \cite{journals/tc/BlaumBBM95,conf/fast/CorbettEGGKLS04,journals/tc/HuangX08,kbp:12:rec}.

Beside the use in massive storage systems, the erasure codes have been recently used in one research area that is addressing the demanding needs for increasing the speed and reliability of packet based communications. That evolving area is Network Coding \cite{Ahlswede2000}. Network Coding allows nodes in the network to perform a set of functions over the generated or received data packets before forwarding them. Random Linear Network Coding (RLNC) \cite{Ho2006} is a network coding technique that produces random linear combinations of the packets over a Galois Field of size $q$, $GF(q)$. The field size has an impact on the decoding probability, i.e., the probability of receiving linearly independent packets increases with $q$.

When one or more sources want to transmit $k$ packets to one or more destination nodes, the channel conditions must be considered. Even in a presence of packet losses (erasures) the destination node has to be able to decode $k$ original packets by receiving $k+r$ packets. The authors in \cite{conf/globecom/LucaniMS09} derive the average number of received coded packets $n$ for successful decoding of $k$ original packets at the destination nodes. They study the effect of $q$ on the performance of RLNC. The exact probability that $k$ out of $k+r$ received packets are linearly independent is derived in \cite{oai:hal.inria.fr:inria-00613337}. Both papers show that $q$ equal to 4 or 8 is enough to get very close to the optimal performance even when $k$ is not very large.

However, as in the case of codes for massive storage systems, working in higher fields or with large number of data packets has an impact on the computational complexity leading to higher energy consumption \cite{conf/icc/HeidePFM11} and no real benefits. A recent result in \cite{pgm:13:sfg} shows that the speed of computation on modern CPUs with wide SIMD instructions is similar for operations in $GF(2)$ and in $GF(16)$. On the other hand, implementing RLNC in higher fields on devices that have power and memory constraints is a challenging problem. Some recent studies show that RLNC in constrained devices in $GF(2)$ is up to two orders of magnitude less energy demanding and up to one order of magnitude faster than RLNC in higher fields \cite{pedersen}, \cite{icc09}.

In this work we introduce a definition of \emph{Families of Optimal Binary Non-MDS Erasure Codes} for $[n, k]$ codes over $GF(2)$. Then we propose one heuristic algorithm for finding those families by using hill climbing techniques over Balanced XOR codes introduced in \cite{conf/ifip6-6/KralevskaGO13}. Due to the hill climbing search, those families of codes have always better decoding probability than the codes generated in a typical Random Linear Network Coding scenario, i.e., random linear codes as described in \cite{oai:hal.inria.fr:inria-00613337}. We also show a surprising result that for small values of $k$, the decoding probability of our codes in $GF(2)$ is very close to the decoding probability of the codes obtained by RLNC but in the higher finite field $GF(4)$.

The paper is organized as follows. In Section \ref{Preliminaries}, we introduce the basic terminology and the definition of Families of Optimal Binary Non-MDS Erasure Codes. In Section \ref{HillClimbing}, we describe one heuristic algorithm for finding those Families of Optimal Binary non-MDS Erasure Codes. We also discuss and compare the properties of our erasure codes to codes generated in a typical Random Linear Network Coding scenario, i.e., random linear codes. Conclusions and future work are summarized in Section \ref{Conclusions}.

\section{Mathematical Preliminaries} \label{Preliminaries}
In this section we briefly introduce the basic terminology, some useful properties and facts about linear codes, as well as some basic terminology and coding methods for Balanced XOR codes \cite{conf/ifip6-6/KralevskaGO13}.

Let us denote by $\mathbf{F}_q = GF (q)$ the Galois field with $q$ elements, and by $\mathbf{F}_q^n$ the $n$-dimensional vector space over $\mathbf{F}_q$. Let us also denote by $[n, k]_q$ the $q$-ary linear code of length $n$ and rank $k$ which is actually a linear subspace $C$ with dimension $k$ of the vector space $\mathbf{F}_q$. An $[n, k, d]_q$ code is an $[n, k]_q$ code with minimum weight at least $d$ among all nonzero codewords. An $[n, k, d]_q$ code is called maximum distance separable (MDS) if $d = n - k + 1$. The Singleton defect of an $[n, k, d]_q$ code $C$ defined as $s(C ) = n - k + 1 - d$ measures how far away is $C$ from being MDS.

Below we give some basic properties for MDS matrices that we use in this paper:
\begin{proposition}[\hspace{-0.025cm}\cite{bok:MW}, Ch. 11, Corollary 3]
Let $C$ be an $[n, k, d]$ code over $GF(q)$. The following statements are equivalent:
\begin{enumerate}
  \item $C$ is MDS;
  \item every $k$ columns of a generator matrix $G$ are linearly independent;
  \item every $n - k$ columns of a parity check matrix $H$ are linearly independent.
\end{enumerate}
\label{MDS_k_columns}
\end{proposition}

\begin{definition}
Let $C$ be an $[n, k]$ code over $GF(q)$ with a generator matrix $G$. Let us denote by $\mathcal{G}_{I}, I=k,\ldots,n$ the sets of submatrices obtained from $G$ when choosing $I$ columns from $G$, and by $\mathcal{D}_{I} \subset \mathcal{G}_{I}, I=k,\ldots,n$ the subsets of $\mathcal{G}_{I}$ with a rank $k$. We call the following vector $V_D=(\varrho_0, \varrho_1, \ldots, \varrho_{n-k})$, $\varrho_{i}=|\mathcal{D}_{i+k}|/|\mathcal{G}_{i+k}|$, the \emph{Vector of Exact Decoding Probability}, for the code $C$.
\label{DecodableSubspaces}
\end{definition}

With other words, the value $\varrho_{i}$ represents the probability that we can decode all $k$ original values $x_1, \ldots, x_k$, if we are given $k+i$ values $y_1, \ldots, y_{k+i}$ that corresponds to encoding with $k+i$ columns of the generator matrix $G$.

For random generator matrices $G$, the values of $V_D$ are calculated in \cite{oai:hal.inria.fr:inria-00613337} and we formulate them in the following Proposition:
\begin{proposition}
For a linear $[n, k]$ code over $GF(q)$ with a random generator matrix $G$ the elements of the vector $V_D=(\varrho_0, \varrho_1, \ldots, \varrho_{n-k})$ have the following values:
\begin{equation}
\varrho_{i} = P(k+i),
\label{rhos}
\end{equation}
where the values $P(I)$ are computed as follows:
\begin{equation}
P(I) = \left\{
\begin{array}{ll}
0 & \text{if } I < k,\\
\prod_{j=0}^{k-1}\Big( 1 - \frac{1}{q^{I-j}}  \Big) & \text{if } I \ge k.
\end{array} \right.
\end{equation}
\label{V-D-RLNC}
\end{proposition}
\begin{proof}
\ The equation (\ref{V-D-RLNC}) is actually the equation (7) in \cite{oai:hal.inria.fr:inria-00613337} with adopted notation to be consistent with the standard notation for linear $[n, k]$ codes over $GF(q)$. The equation (\ref{rhos}) then follows directly.
\end{proof}

The connection between the Vector of Exact Decoding Probability and the MDS codes can be established by using the Proposition \ref{V-D-RLNC} as follows:
\begin{theorem}
A linear $[n, k]$ code $C$ over $GF(q)$ with a generator matrix $G$ is a MDS code iff the Vector of Exact Decoding Probability is the following vector $V_D=(\varrho_0, \varrho_1, \ldots, \varrho_{n-k})=(1, 1, \ldots, 1)$.
\end{theorem}
\begin{proof}
\ The theorem can be proved with a direct application of the Proposition \ref{V-D-RLNC} and the Definition \ref{DecodableSubspaces}.
\end{proof}

In this work we are interested exclusively to work with XOR coding, i.e., to work with linear binary codes. Thus, our interest is to define a class of binary codes that in some properties are as close as possible to MDS codes. Unfortunately, it is a well known old fact in coding theory (see for example \cite{bok:MW}) that for the case of linear binary codes, all MDS codes are trivial, i.e., $k=1$ or $n=k+1$ or $n=k$.

So, dealing with the fact that non-trivial binary codes are not MDS, we adopt a strategy to search for codes that will be optimal from certain perspective according to the Vector of Exact Decoding Probability $V_D$. When a channel has an erasure probability $p$ the strategy will be to find binary codes that maximize the probability to recover the original data. Therefore, we prove the following Theorem:
\begin{theorem}
Let $C$ be a binary linear $[n, k]$ code with a Vector of Exact Decoding Probability $V_D=(\varrho_0, \varrho_1, \ldots, \varrho_{n-k})$ and let $k$ packets are encoded by $C$. The probability $p_s$ of successful decoding of $k$ packets from $n$ encoded and transmitted packets via a channel with an erasure probability $p$ is:
\begin{equation}
\footnotesize
p_s = 1 - \Bigg( \sum_{i=0}^{n-k} \binom{n}{i} p^{i}(1-p)^{n-i} (1 - \varrho_{n-k-i}) + \sum_{i=n-k+1}^{n} \binom{n}{i} p^{i}(1-p)^{n-i} \Bigg)
\label{SuccessfulDecodingProbability}
\end{equation}
\end{theorem}
\begin{proof}
\ Let us denote by $E_1$ the event that $i$ packets, where $0\leq i\leq n-k$, are lost during the transmission, and by $E_2$ the event that more than $n-k$ packets from the set of all $n$ packets are lost during the transmission.

The probability of the event $E_1$ is calculated by the expression:
\begin{equation}
P(E_1) = \sum_{i=0}^{n-k} \binom{n}{i} p^{i}(1-p)^{n-i},
\label{failure1}
\end{equation}
and the probability of the event $E_2$ is:
\begin{equation}
P(E_2) = \sum_{i=n-k+1}^{n} \binom{n}{i} p^{i}(1-p)^{n-i}.
\label{failure2}
\end{equation}

From expression (\ref{failure1}) we compute the probability $p_{u_1}$ of failure to decode $k$ original packets, by multiplying every value in the sum by the opposite probability of successful decoding when $n-k-i$ columns of the generator matrix $G$ are received, i.e., when $i$ packets are lost. So the decoding failure probability if $i$ packets are lost ($0\leq i\leq n-k$) is computed by the following expression:
\begin{equation}
p_{u_1} = \sum_{i=0}^{n-k} \binom{n}{i} p^{i}(1-p)^{n-i}(1 - \varrho_{n-k-i}).
\label{failure3}
\end{equation}

If more than $n-k$ packets are lost then the probability to fail the decoding is 100\% thus the probability $p_{u_2}$ of failure to decode $k$ original packets is equal to $P(E_2)$, i.e., $p_{u_2}=P(E_2)$.

In total, the probability of unsuccessful decoding $p_u$ is:
\begin{eqnarray}
p_u&=& p_{u_1} + p_{u_2} = \\  \nonumber
&=&\sum_{i=0}^{n-k} \binom{n}{i} p^{i}(1-p)^{n-i} (1 - \varrho_{n-k-i}) + \\ \nonumber
& & + \sum_{i=n-k+1}^{n} \binom{n}{i} p^{i}(1-p)^{n-i} \\  \nonumber
\label{failure4}
\end{eqnarray}

Finally the probability $p_s$ of successful decoding of $k$ packets is the opposite probability of $p_u$ i.e.,
$$p_s = 1 - p_u .$$
\end{proof}

Having defined the probability $p_s$ of successful decoding of $k$ packets that are encoded with an $[n, k]$ binary code, we define a \emph{Family of Optimal Binary Non-MDS Erasure Codes} as follows:
\begin{definition}
Let $\mathcal{C}$ be a family of binary linear $[n, k]$ codes that have a probability $p_s$ of successful decoding $k$ packets from $n$ encoded and transmitted packets via a channel with an erasure probability $p$. We say that $\mathcal{C}$ is a \emph{Family of Optimal Binary Non-MDS Erasure Codes} if for every binary linear $[n, k]$ code $C'$ with a probability $p'_s$ of successful decoding of $k$ packets in a channel with an erasure probability $p$, there exist a code $C \in \mathcal{C}$ with a probability $p_s$ of successful decoding, such that $p'_s \le p_s$, for every erasure probability $p$.
\end{definition}
\begin{problem}
For given values of $n$ and $k$ find a Family $\mathcal{C}$ of Optimal Binary Non-MDS Erasure Codes.
\label{problem1}
\end{problem}

\section{A Hill Climbing Heuristics For Finding Families of Optimal Binary Non-MDS Erasure Codes} \label{HillClimbing}

Finding exact analytical solution (or finding deterministic and efficient algorithm that will find the solution) for the Problem \ref{problem1} is hard and in this moment we do not know such a solution. However, there are many heuristic optimization methodologies that can be used for a search of approximate solutions. We choose to use the simplest one: The Stochastic Hill-Climbing Methodology\cite{book:russell:2003}. The hill climbing heuristics has been already used in optimizing problems for RLNC such as in \cite{6012235}. In general, the stochastic  heuristics is defined as in Algorithm 1.
\begin{table}
\centering
\caption{A general Stochastic Hill-Climbing algorithm for finding a Family of Optimal Binary Non-MDS Erasure Codes for given values of $n$ and $k$}\label{GeneralStochasticHillClimbing}
\begin{tabular}{|c|}
  \hline
   \parbox{6.0cm}{\center \textbf{Algorithm 1} \vspace{0.1cm}}\\ %
  \hline
  \parbox{6.0cm}{\vspace{0.1cm} {\bf Input.} $n$ and $k$ \vspace{0.1cm}}\\
  \hline
  \parbox{6.0cm}{\vspace{0.1cm} {\bf Output.} A candidate Family $\mathcal{C}$ of Optimal Binary Non-MDS Erasure Codes \vspace{0.1cm}}\\
  \hline
  \begin{tabular}{l}
  \parbox{6.0cm}{\vspace{0.1cm} 1. Find a random $[n, k]$ linear binary code and compute its Vector of Exact Decoding Probability $V_D=(\varrho_0, \varrho_1, \ldots, \varrho_{n-k})$ and its   probability $p_s$ of successful decoding of $k$ packets from the equation (\ref{SuccessfulDecodingProbability}). }\\
  \parbox{6.0cm}{\vspace{0.1cm} 2. Repeatedly improve the solution until no more improvements are necessary/possible.}\\
 \end{tabular}\\
 \hline
  \end{tabular}
\vspace{-0.5cm}
\end{table}

In order to improve the codes found by Algorithm 1 we decided to work with balanced structures as they were introduced in \cite{conf/ifip6-6/KralevskaGO13}.

\begin{definition}
A XOR-ed coding is a coding that is realized exclusively by bitwise XOR operations between packets with equal length. Hence, it is a parallel bitwise linear transformation of $k$ source bits $x=(x_1,\ldots,x_k)$ by a $k \times k$ nonsingular binary matrix $\mathbf{K}$, i.e., $y = x \cdot \mathbf{K} $.
\end{definition}

In other words XOR-ed coding assumes work within the smallest finite field $GF(2)$, i.e., with $k \times k$ nonsingular binary matrices $\mathbf{K}$. While the binary matrices $\mathbf{K}$ in general can be of any form, the specifics about matrices introduced in \cite{conf/ifip6-6/KralevskaGO13} are that they are highly structured, balanced and their construction is based on Latin rectangles of dimensions $k_1 \times k$.

\begin{definition}
A Latin square of order $k$ with entries from an $k$-set X is an $k\times k$
array $L$ in which every cell contains an element of X such that every row of $L$ is a
permutation of X and every column of $L$ is a permutation of X.
\label{Latinsquare}
\end{definition}

\begin{definition}
A $k_1\times k$ Latin rectangle is a $k_1\times k$ array (where $k_1 \leq k$) in which each cell contains
a single symbol from an $k$-set X, such that each symbol occurs exactly once in each
row and at most once in each column.
\label{Latinrectangle}
\end{definition}

\begin{definition}
Let $(X,A)$ be a design where $X = \{x_1,\ldots, x_v\}$ and $A = \{A_1,\ldots,A_b\}$. The incidence matrix of $(X,A)$ is the $v\times b$ 0-−1 matrix $M = (m_{i,j})$ defined by the rule
$
m_{i,j} =
\begin{cases}
1, & \text{if}\  x_i\in A_j, \nonumber\\
0, & \text{if}\  x_i\notin A_j. \nonumber\\
\end{cases}
$
\label{incidencematrix}
\end{definition}

\begin{proposition}[\cite{conf/ifip6-6/KralevskaGO13}]
The incidence matrix $M = (m_{i,j})$ of any Latin rectangle with dimensions $k_1\times k$ is a balanced matrix with $k_1$ ones in each row and each column.
\label{balanced}
\end{proposition}

\begin{proposition}[\cite{conf/ifip6-6/KralevskaGO13}]
The necessary condition an incidence matrix $M = (m_{i,j})$ of a $k_1\times k$ Latin rectangle to be nonsingular in $GF(2)$ is $k_1$ to be odd, i.e., $k_1=2 l + 1$.
\label{nonsingular}
\end{proposition}
\begin{example}
\label{Example01}
Let us take the following Latin square and split it into two Latin rectangles:
$$\small
L =
\begin{bmatrix}
 1 & 4 & 3 & 5 & 2 \\
 3 & 1 & 5 & 2 & 4 \\
 4 & 2 & 1 & 3 & 5 \\
 \hline
 5 & 3 & 2 & 4 & 1 \\
 2 & 5 & 4 & 1 & 3 \\
\end{bmatrix}.
$$
The incidence matrix $M$ of the $3\times 5$ upper Latin rectangle is:
$$\small
M =
\begin{bmatrix}
 1 & 0 & 1 & 1 & 0 \\
 1 & 1 & 0 & 1 & 0 \\
 1 & 0 & 1 & 0 & 1 \\
 0 & 1 & 1 & 0 & 1 \\
 0 & 1 & 0 & 1 & 1 \\
\end{bmatrix}.
$$

Note how balanced are the rows and columns: in every row and every column, the number of 1s is 3.
\end{example}

The following proposition follows directly from the Proposition \ref{nonsingular}:
\begin{proposition}
The $k+1$-th column of the generator matrix $G$ of a trivial $[k+1, k]_2$ MDS code that has in the first $k$ columns a matrix for a balanced XOR-ed coding consists of all 1s.
\label{trivialXORMDS}
\end{proposition}

We now describe the modified Stochastic Hill-Climbing that is using Balanced XOR codes where one column of the generator matrix is defined as in Proposition \ref{trivialXORMDS}:
\begin{table}[!h]
\centering
\caption{A Stochastic Hill-Climbing algorithm for finding a Family of Optimal Binary Non-MDS Erasure Codes based on Balanced XOR codes}\label{BNCStochasticHillClimbing}
\begin{tabular}{|c|}
  \hline
   \parbox{6.0cm}{\center \textbf{Algorithm 2} \vspace{0.1cm}}\\ %
  \hline
  \parbox{6.0cm}{\vspace{0.1cm} {\bf Input.} $n$ and $k$ \vspace{0.1cm}}\\
  \hline
  \parbox{6.0cm}{\vspace{0.1cm} {\bf Output.} A candidate Family $\mathcal{C}$ of Optimal Binary Non-MDS Erasure Codes \vspace{0.1cm}}\\
  \hline
  \begin{tabular}{l}
  \parbox{6.0cm}{\vspace{0.1cm} 1. Find a random Balanced XOR code and put it as the first part of the generator matrix $G$ of an $[n, k]$ code. Set the $k+1$-th column to consists of all 1s, and set the remaining columns with random values. Compute the Vector of Exact Decoding Probability $V_D=(\varrho_0, \varrho_1, \ldots, \varrho_{n-k})$ and its probability $p_s$ of successful decoding of $k$ packets from the equation (\ref{SuccessfulDecodingProbability}). }\\
  \parbox{6.0cm}{\vspace{0.1cm} 2. Repeatedly improve the solution until no more improvements are necessary/possible.}\\
 \end{tabular}\\
 \hline
  \end{tabular}
\end{table}

We would like to note that Algorithm 1 can find codes with similar decoding probabilities as Algorithm 2, but after performing more stochastic search attempts. Moreover, the codes that Algorithm 2 finds have advantages that they are structured, balanced and they are sparse, where the sparsity can go down to just 3 nonzero positions. 

We now give two numerical results that compare the performance of our codes to a typical linear random code in $GF(2)$ that can be generated in RLNC. The same parameters are taken as in \cite{oai:hal.inria.fr:inria-00613337}, i.e., $r={0, \ldots , 8}$ is the number of excess packets for $k=5$ and $k=100$. The results show that the decoding probability with our scheme is closer to the decoding probability under RLNC in $GF(4)$ when $k$ is small. We would like to emphasize that with Algorithm 2 we could easily find codes with $k$ in range $[5, \ldots, 1000]$.
\begin{figure}
\begin{minipage}[b]{0.5\linewidth}
\centering
\includegraphics[width=3.4in]{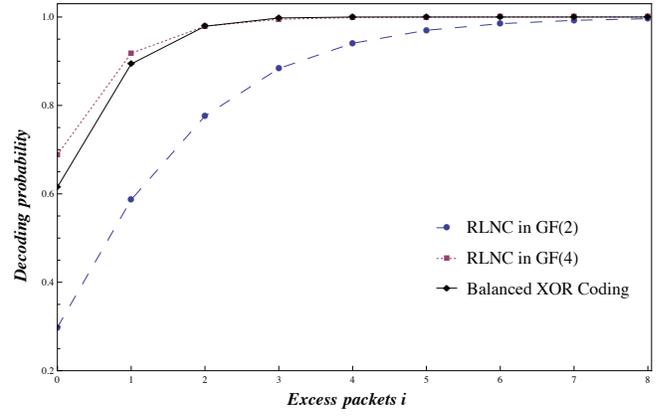}
\end{minipage}
\caption{Vector of Exact Decoding Probability $V_D$ for $k$=5}
\label{k_5}
\end{figure}

In Figure \ref{k_5} the code that was found after 10,000 stochastic attempts by the Balanced XOR-ed approach of Algorithm 2 is based on the Latin Square from Example \ref{Example01}. Its generator matrix is the following:
$$\small
G =
\left[
\begin{array}{ccccccccccccc}
 1 & 1 & 1 & 0 & 0 & 1 & 0 & 0 & 0 & 0 & 1 & 0 & 1 \\
 0 & 1 & 0 & 1 & 1 & 1 & 0 & 0 & 1 & 0 & 1 & 1 & 0 \\
 1 & 0 & 1 & 1 & 0 & 1 & 0 & 1 & 0 & 0 & 0 & 1 & 0 \\
 1 & 1 & 0 & 0 & 1 & 1 & 0 & 0 & 0 & 1 & 0 & 1 & 0 \\
 0 & 0 & 1 & 1 & 1 & 1 & 1 & 0 & 0 & 0 & 1 & 0 & 0 \\
\end{array}
\right]
$$
The Vector of Exact Decoding Probability for this code is: $V_D=(0.615, 0.895, 0.979, 0.998, 1., 1., 1., 1.)$ and is presented in Figure \ref{k_5} with a solid line.

A typical random linear code in $GF(2)$ generated in RLNC is presented in Figure \ref{k_5} with a dashed line. For comparison purposes, we put the values for decoding probabilities of a typical random linear code in $GF(4)$ in the same Figure \ref{k_5}. As it can be seen, our codes in $GF(2)$ have decoding probabilities as a random linear code in $GF(4)$.

The real advantage of our codes is seen in Figure \ref{comparison_13_5} in channels where packet losses occur with certain probabilities. Similarly as in \cite{oai:hal.inria.fr:inria-00613337} we give the results for $[n, k]=[108,100]$ in Figure \ref{k_100} and in Figure \ref{comparison_108_100}.
\begin{figure}
\begin{minipage}[b]{0.5\linewidth}
\centering
\includegraphics[width=3.4in]{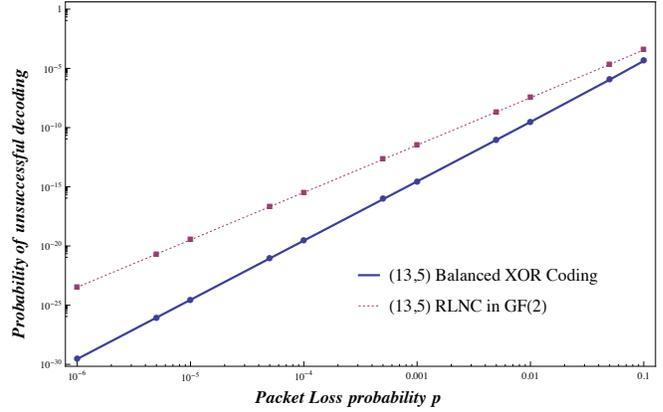}
\end{minipage}
\caption{Comparison between probabilities of unsuccessful decoding of a typical RLNC code and a code obtained with our stochastic strategy in $GF(2)$ for $k=5$}
\label{comparison_13_5}
\end{figure}

\begin{figure}
\begin{minipage}[b]{0.5\linewidth}
\centering
\includegraphics[width=3.4in]{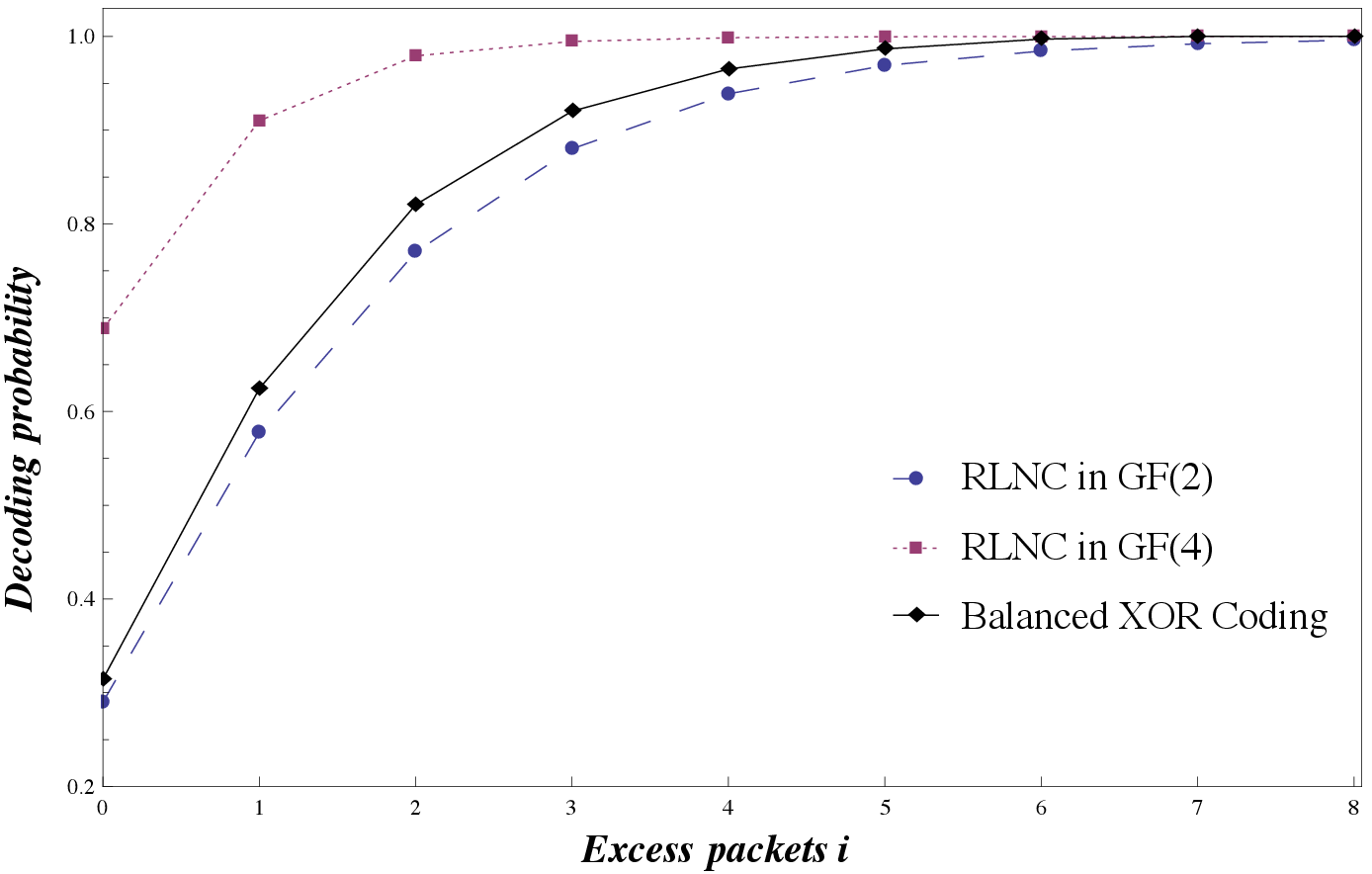}
\end{minipage}
\caption{Vector of Exact Decoding Probability $V_D$ for $k$=100}
\label{k_100}
\end{figure}
\begin{figure}
\begin{minipage}[b]{0.5\linewidth}
\centering
\includegraphics[width=3.4in]{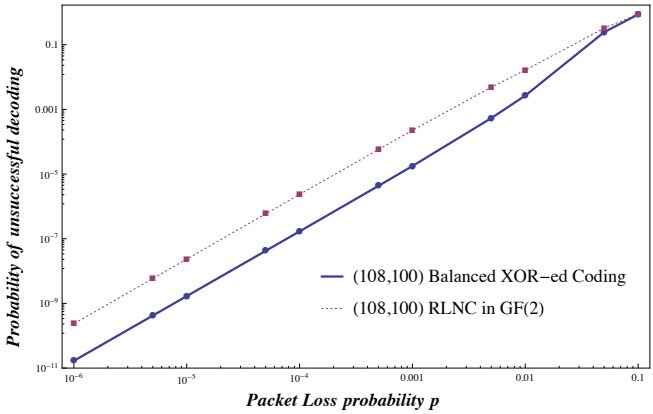}
\end{minipage}
\caption{Comparison between probabilities of unsuccessful decoding of a typical RLNC code and a code obtained with our stochastic strategy in $GF(2)$ for $k=100$}
\label{comparison_108_100}
\end{figure}

\section{Conclusions}\label{Conclusions}
We introduced a definition of \emph{Families of Optimal Binary Non-MDS Erasure Codes} for $[n, k]$ codes over $GF(2)$ and proposed one heuristic algorithm for finding those families using hill climbing techniques over Balanced XOR codes. We showed that the families of codes that we found have always better decoding probability than the decoding probability of random linear codes generated in RLNC. We also showed that for small values of $k$ the decoding probability of our codes in $GF(2)$ is very close to the decoding probability of the random linear codes in $GF(4)$.

As a next research direction, we point out that it will be very useful to further investigate the theoretical lower and upper bounds of decoding probabilities of the defined Families of Optimal Binary Non-MDS Erasure Codes and to find better heuristic or deterministic algorithms for efficient finding of those families. It would be a natural research directions to see how this methodology performs in higher fields.

\section*{Acknowledgements}\label{Ack}
We would like to thank Harald {\O}verby and Rune E. Jensen for their discussions that improved the quality of this paper. We would also like to thank the anonymous reviewers for their useful comments and suggestions.

\bibliographystyle{plain}

\end{document}